\documentclass[11pt]{article}
\usepackage[utf8]{inputenc}

\usepackage{amsmath, amsthm}
\usepackage{thmtools}
\usepackage{thm-restate}

\usepackage[margin=3.5cm]{geometry}

\usepackage{amssymb}

\newtheorem{theorem}{Theorem}
\newtheorem{corollary}[theorem]{Corollary}

\newtheorem{lemma}[theorem]{Lemma}

\newcommand{\items}{\mathcal{M}}

\title{A lone divider allocation algorithm with subjective divisibility}

\author{Uriel Feige\thanks{Weizmann Institute, Israel. {\tt uriel.feige@weizmann.ac.il}}}

\begin{document}

\maketitle

\begin{abstract}
    In the subjective divisibility allocation model, all goods are divisible, every agent $i$ has a non-negative additive valuation function $v_i$, where for some goods $v_i$ is additive over fractions of the good, and for some goods, any fraction smaller than~1 has value~0. Previous work established the existence of $\frac{5}{9}$-MMS allocations in this setting. We show that $\frac{3}{5}$-MMS allocations exist. Our proof is based on an adaptation of the lone divider framework.
\end{abstract}

\section{Introduction}

We consider a setting of allocating items with subjective divisibility, as defined and studied in~\cite{BeiLL25}. There are $n$ agents and a set $\items$ of $m$ divisible goods, also referred to as items. An allocation is represented as an $n$ by $m$ matrix $A$ that assigns to every agent $i$ a fraction $A_{ij}$ of item $j$, with $A_{ij} \ge 0$ and $\sum_i A_{ij} \le 1$ for every $j$.  Row $A_i$ is the allocation to agent $i$. Each agent $i$ has a set $D_i \subseteq \items$, where every item $j \in D_i$ is divisible for $i$, and every item $j \not\in D_i$ is indivisible for $i$. Each agent has a non-negative valuation function $v_i$ over the items. The value of an allocation $A_i$ for agent $i$ is:

$$v_i(A_i) = \sum_{\{j \mid A_{ij} = 1\}} v_i(j) \; + \sum_{\{j \in D_i \mid A_{ij} < 1\}} A_{ij} \cdot v_i(j)$$

The maximin share (MMS) value for agent $i$ is:

$$MMS(\items, v_i, D_i, n) = \max_A \min_j v_i(A_j)$$

In~\cite{BeiLL25} it was shown that there always are allocations giving each agent at least $\frac{1}{2}$-MMS, and for $n \le 3$, there always are allocations giving each agent at least $\frac{2}{3}$-MMS. In~\cite{BeiDLW26} the bounds were improved, showing that there always are allocations giving each agent at least $\frac{5}{9}$-MMS, and for $n \le 4$, there always are allocations giving each agent at least $\frac{2}{3}$-MMS. For every $n \ge 2$, there are allocation instances in which every allocation gives at least one of the agents no more than $\frac{2}{3}$-MMS. In this note, we improve the existential results for $n \ge 5$.

\begin{theorem}
\label{thm:main}
In the subjective divisibility setting defined above, there always are allocations giving each agent at least $\frac{3}{5}$-MMS. 
\end{theorem}

Our allocation algorithm proving Theorem~\ref{thm:main} is based on the lone divider approach. This approach was previously shown to give $(\frac{2}{3} + \Theta(\frac{1}{n}))$-MMS allocations for instances with indivisible items and additive valuations~\cite{procaccia2014fair}. More related to our work, it was also adapted to a setting with sellable items~\cite{FeigeGafni2026}, where it was shown to give $\frac{2}{3}$-MMS allocations. Here we adapt it to the setting with subjective divisiblity. The difficulty is that the lone divider might choose to divide items that other agents view as indivisible. This may result in creation of parts in the division that other agents can neither accept (the value is too low) nor let go of (as then only worthless fractions of some indivisible items remain). 

To overcome the above difficulty, we require the division to be {\em pure}: no previously undivided item is allowed to be subdivided by the lone divider. This requirement was inspired by a related requirement imposed in~\cite{FeigeGafni2026}. The problem then becomes that insisting on pure divisions reduces the approximation ratio. To counter this, we add preliminary steps before we apply the lone divider algorithm. One type, that we refer to as {\em one-agent} steps, was already used in previous algorithms in the subjective divisibility setting, and in the sellable item setting. With such a preliminary step, the lone pure divider approach guarantees $\frac{1}{2}$-MMS.  

Improving the approximation ratio from $\frac{1}{2}$ to {$\frac{3}{5}$} is our main contribution. We use a combination of two new ideas. One is another type of preliminary step, referred to as {\em two-agent} step. The other is that we allow the lone divider to fail in producing a pure division, but we show that if this happens (and given that we performed the preliminary steps), an allocation can be found using a different procedure, referred to as a {2-matching}.

\section{Our $\frac{3}{5}$-MMS allocation algorithm}

Consider an arbitrary allocation instance with subjective divisibility. For every agent $i$, we assume that the valuation function $v_i$ is scaled so that the MMS value is~1. For simplicity in the analysis, we may further assume that the value of each part is exactly~1, by pretending (in the analysis, not in the run of the algorithm) 
that some items have value smaller than their actual value, if needed. 
Hence, initially, the total value of all items for each agent is $n$. 
(For each agent $i$ we pretend that there is an ``auxiliary" valuation function that initially equals $v_i$. Then, during the analysis, values of items according to the auxiliary valuation may decrease to below their value under $v_i$. All value statements in our proofs refer to this auxiliary valuation. A bundle proved acceptable under the auxiliary valuation is certainly acceptable under $v_i$.)

For a given value of $\rho$, we say that a bundle $B$ (containing items and fractions of items) is {\em acceptable} for agent $i$ if $v_i(b) \ge \rho$.
We present a {polynomial time} 
allocation algorithm that given valuation functions scaled so that the MMS value is~1, outputs a $\rho$-MMS allocation. The algorithm will proceed in steps, where in each step, some agents receive acceptable bundles and leave, and the other agents remain active. As to items, some will remain {\em free} (available for active agents in future steps), some will become {\em allocated} (if already allocated in full), and some will become {\em partial items} (if some fraction of them was allocated, the fraction that remains unallocated is a partial item). An item that is either free or partial will be referred to as {\em available}.

Our algorithm applies the following types of steps, moving to the next type of step only after the current type of step is no longer applicable.  
We will show that for our choice of $\rho = \frac{3}{5}$, every agent receives an acceptable bundle.

\begin{enumerate}

\item {\em One-agent} steps. A single active agent receives (part of) a single item, and leaves.

As long as there is an active agent for which one of the available items is acceptable, allocate the smallest fraction of an item that is acceptable for some agent, to that agent. 

\item {\em Two-agent} steps. Two active agents each receives (parts of) a set of three items, and leave.

As long as there are two active agents for which three of the available items can be partitioned into two bundles, one acceptable for one agent and the other to the other agent, do such an allocation. Among such allocations, choose the one that has the smallest fraction $f$ in the following sense. Three available items, say $\{e_1,e_2,e\}$ are allocated. Items $e_1$ and $e_2$ are allocated in full, each to a different agent. Item $e$ is referred to as the {\em shared item}. An $f$ fraction of $e$ is divided among the two agents, and a $1 - f$ fraction of $e$ remains available for future steps. Having fixed the set $\{e_1,e_2,e\}$, the choice of two agents needs to be such that $f$ is as small as possible,

\item {\em Lone pure divider} steps. Several active agents receive acceptable bundles and leave.

Let $n'$ denote the number of active agents remaining. Call a partition of the available items (free and partial) {\em pure} if no free item is split into more than one part. If some agent $i$ has a pure partition of the remaining items into $n'$ parts acceptable for $i$ (as constructed by the polynomial time procedures in the proofs of Lemmas~\ref{lem:3divisible} and~\ref{lem:4large}), then we execute the lone divider pure step, using this partition. Construct a bipartite graph with the $n'$ active agents on one side, the $n'$ parts on the other, and edges between active agents and the  parts that they find acceptable. As in other applications of the lone divider approach, there are two possibilities. If there is a perfect matching (every agent is matched to a different part), then the allocation algorithm ends. If no matching is perfect, then there is an ``envy-free matching": a nonempty subset of active agents each gets an acceptable part and leaves,  whereas for each of the remaining agents, each allocated part is not acceptable.

\item {\em 2-matching} step. Each one of the active agents receives two of the available items and leaves. The allocation algorithm ends. 

We say that an available item $e$ (for partial items, $e$ signifies only the part of the item that is available) is {\em special} for agent $i$ if $1 - \rho <  v_i(e) < \rho$ (the inequalities make sense, because $\rho > \frac{1}{2}$). If for each active agent there are two available items that are special only for that agent, each active agent receives two special items, and the algorithm ends. The bundle received by such an agent is acceptable if $2(1 - \rho) \ge \rho$, which holds for $\rho \le \frac{2}{3}$.
  
\end{enumerate}


\subsection{Analysis of the approximation ratio}

We bound from below the value of $\rho$ that guarantees successful termination of our allocation algorithm. In every step, the agents that leave receive acceptable bundles. The key to the analysis is to bound from below the value that the remaining available items have for an active agent.

Fix an active agent $i$ and its given MMS partition (in which each part has value~1). We analyse the effect of each type of step on the value that remains for $i$. Our initial goal will be to show that the loss incurred by each leaving agent is at most $2 - 2\rho$. (As one example of where the expression $2(1 - \rho)$ value comes up in our analysis, see the 2-matching step described above.)

{\em One-agent steps, indivisible item}. Suppose that an agent left with item $e$ that agent $i$ views as indivisible. If $e$ was a partial item, no value was lost (as parts of $e$ have no value for $i$). If $e$ was a free item, then only the MMS part containing $e$ lost value, whereas the number of agents was reduced by one. Hence, the MMS value for $i$ did not decrease, and for agent $i$ we set the new value of $n$ to be one less that its previous value. (For simplicity of the analysis, pretend that the values of other items and fractional parts of items in the bundle of $e$ are dropped to~0.)

From now on, we shall interpret the value of $n$ to be the updated value, after taking into account all reductions in the value of $n$ that occurred due to ``one-agent steps, indivisible item" steps, and assume that no such steps occur. (Under this interpretation, different active agents may see different values of $n$.)

{\em One-agent steps, divisible item}. Suppose that an agent left with item $e$ that agent $i$ views as divisible. Agent $i$ lost a value of at most $\rho$, as otherwise $i$ herself was willing to leave with a smaller fraction of $e$. The inequality $\rho \le 2 - 2\rho$ holds for $\rho \le \frac{2}{3}$.

{\em Two-agent steps, divisible shared item.} Suppose that two agents left with three items, and the shared item $e$ is divisible for agent $i$. Then each of the leaving agents leaves with a bundle that $i$ values as at most $\rho$, as otherwise $i$ herself was willing to leave instead of the leaving agent, and with a smaller fraction of $e$. Hence, $i$ loses a value of at most $2\rho \le 2(2 - 2\rho)$, where the inequality holds for $\rho \le \frac{2}{3}$.

{\em Two-agent steps, indivisible shared item.} Suppose that two agents left with three items, and the shared item $e$ is indivisible for agent $i$. If $e$ was a partial item, then agent $i$ lost no value due to $e$ being taken, and the two-agent step is equivalent from her perspective to two separate one-agent step. Thus we may assume without loss of generality that $e$ was a free item.
Then, agent $i$ loses a value of at most $3\rho$ (one $\rho$ for each item). To simplify future analysis (and without loss of generality), we pretend that the value of the lost shared item $e$ was highest possible, $\rho$, and that the value of other items or fractions of items in the MMS part that contained $e$ is reduced so that their total value of $1 - \rho$.

To partly compensate for the relatively large loss of ``two-agent steps, indivisible shared item", we rearrange the MMS partition of agent $i$. We keep track of a parameter $n_i$, whose initial value is $n$. Each time that the above type of event happens, $n_i$ is decreased by one, meaning that one MMS part is eliminated. The part that is eliminated is the one containing the shared item $e$. Its surplus value $1 - \rho$ is not discarded, but rather, compensates for some of the value lost by the other two items that were allocated. Overall, among the two agents leaving, one is accounted for by reducing $n_i$ by one, and only for one leaving agent, we need to account for the value lost. This value is at most $2\rho - (1- \rho) = 3\rho - 1$. We note two relevant bounds:

\begin{enumerate}
    \item If $\rho \le \frac{2}{3}$, the value lost is at most~1.
    \item If $\rho \le \frac{3}{5}$, the value lost is at most $\le 2 - 2 \rho$.
\end{enumerate}

{\em Lone pure divider} step.  Each part taken by an agent that leaves has value smaller than $\rho$ for $i$. As no free item is partitioned in this step, agent $i$ did not suffer additional loss due to partitioning of an item that she views as indivisible. Hence, again, the loss per agent is at most $\rho$, which is at most $2 - 2\rho$ for $\rho \le \frac{2}{3}$.

\subsection{Establishing that eventually, no agent is active}

It remains to show that if the lone pure divider step is not applicable (and likewise, the earlier one-agent and two-agent steps are not applicable), then the 2-matching allocation step is applicable.

\begin{lemma}
\label{lem:3divisible}
Suppose that $\rho \le \frac{3}{5}$, neither a one-agent nor a two-agent step can be executed, and let $n'$ denote the number of active agents. If for active agent $i$ 
at most one available item (free or partial) is both divisible and special (of value more than $1 - \rho$), then $i$ has a pure partition as required for the lone pure divider step.
\end{lemma}

\begin{proof}
Under the conditions of the lemma, we need to show that a pure partition into $n'$ parts exists, with each part of $v_i$ value at least $\rho$. Observe that for $\rho \le \frac{2}{3}$, the total value of all available items to $i$ is at least $n'$.

Consider first the set $V$ of available free items that $i$ views as both indivisible and special (of value larger than $1 - \rho$, where note that $1 - \rho \ge \frac{1}{3}$ for $\rho \le \frac{2}{3}$). If $|V| \ge 2n'$, then $i$ can place in each part of the partition two such items, achieving a pure partition as desired (for $\rho \le \frac{2}{3}$). Hence, we may assume that $|V| < 2n'$. Construct a graph $G(V,E)$ whose vertex set is $V$ as above, and two vertices are joined by an edge if the sum of values of the respective items is at most~1. Let $M$ be a maximum matching in this graph, and let $|M|$ denote the number of edges in this matching.

Suppose first that $|V| - |M| \le n'$. Then create $|M|$ parts, each containing the two vertices of one edge from $M$. Each such part has value at most~1 and above $\rho$. Hence, no additional items need to be placed in these parts. It remains to construct $n' - |M|$ additional parts. Each unmatched item from $V$ is placed in a different such parts. This can be done because we assume that $|V| - |M| \le n'$.

As long as some part has value lower than $\rho$, we do the following. If there are free indivisible items left, then the value of each such item is at most $1 - \rho$. 
As long as there are free indivisible items left, add them to parts of value smaller than $\rho$ (as long as there are such parts). No part can exceed a value of~1. Now, as long as there are non-special divisible items left (free or partial, of value at most $1 - \rho$, add them to parts of value smaller than $\rho$ (as long as there are such parts). No part can exceed a value of~1. Finally, there may be one special divisible item of value more than $1 - \rho$, but at most $\rho$. Place it in a part of value lower than $\rho$, if such parts exist.

The above partition is pure. We show that each part has value at least $\rho$. At most one part has value above~1 (the part containing the special divisible item), but at most $2\rho$. Suppose for the sake of contradiction that some part has value below $\rho$. Then the sum of values of the two parts is below $3\rho$. As $\rho \le \frac{2}{3}$, this value is strictly below~2. This means that the total value of all parts is strictly smaller than $n'$. contradicting the fact that it is at least $n'$. 

It remains to consider the case that $|V| - |M| > n'$. Here we make use of the notion of $n_i$ introduced in the ``two-agent steps, indivisible shared item" steps. Let $t_i$ be the number of such steps, and so $n_i = n - t_i$. Consider the MMS partition of agent $i$ into $n_i$ parts (those parts not containing the allocated indivisible shared items). Observe also that no item of the other original MMS parts is in $V$, because the value left in each such part is $1 - \rho$. Recall that each part has value~1. As each item in $V$ has value above $\frac{1}{3}$, no part can contain three items from $V$. As $M$ is a maximum matching, at most $|M|$ of the $n_i$ parts have two items from $V$, implying that $n_i \ge |V| - |M|$. 

The total value available to $i$ is computed as follows. There were $t_i$ ``two-agent steps, indivisible shared item" steps. In each such step, two agents left and the value lost was at most $3\rho$. In addition, $n - n' - 2t_i$ agents left, each reducing the value by $\rho$. So the total value left is $n' + (n - n')(1 - \rho) - t_i \rho$.

Following the same proof as in the case that $|V| - |M| \le n'$, the difference is that $|V|-|M| - n' \le n_i - n'$ parts now need to contain two unmatched items from $v$, rather than only one unmatched item from $v$. Each such part might have value as high as $2\rho$ instead of just~1, so there is wasted value of at most $(n_i-n')(2\rho - 1) = (n - t_i - n')(2\rho - 1)$. 

In the total value left there was a surplus of at least $(n - n')(1 - \rho) - t_i \rho$ (we do not claim at this point that the surplus is non-negative), so the net surplus (after subtracting the wasted value) is: 

$$(n - n')(1 - \rho) - t_i \rho - (n - t_i -n')(2\rho - 1) = (n - n')(2 - 3\rho) - t_i(1 - \rho) \ge t_i(3 - 5\rho)$$

In the last inequality we used the fact that $t_i \le \frac{n - n'}{2}$.

The surplus is nonnegative when $\rho \le \frac{3}{5}$. Thereafter, the remaining parts of the proof work without change.
\end{proof}

\begin{lemma}
\label{lem:4large}
Suppose that $\rho \le \frac{2}{3}$, and let $n'$ denote the number of active agents. If for active agent $i$ the total value of available items is at least $n'$, no remaining item has value above $\rho$, and there is at most two available items (free or partial) of value above $\frac{1}{2}$, then $i$ has a pure partition as required for the lone divider pure step.
\end{lemma}

\begin{proof}
Suppose that there are only two remaining items of value above $\frac{1}{2}$. (The case where there is only one such item is simpler and omitted.) We construct a pure partition into $n'$ parts. Place both  these items in the first part (giving value above $\rho$).
Order the remaining items in decreasing order of value, and place them one by one arbitrarily in other parts that have value less than $\rho$ (as long as there are such parts). As no remaining item has value above $\frac{1}{2}$, no new part can overflow if $\rho \le \frac{2}{3}$. (Any part of value above $\frac{1}{2}$ must contain at least two items. If the part has value below $\frac{2}{3}$, one item  must have value below $\frac{1}{3}$. Having ordered the items, future items have value below $\frac{1}{3}$.) 
Thus, the total overflow is at most $2\rho - 1 \le 1 - \rho$ from the first part, where the inequality holds for $\rho \le \frac{2}{3}$. Thus, no part can have deficit larger than $1 - \rho$, implying that every part has value at least $\rho$.
\end{proof}

\begin{corollary}
\label{cor:fractional}
    For $\rho \le \frac{3}{5}$, if no one-agent, two-agent and lone pure divider step is applicable, then the 2-matching step is applicable (after which the allocation algorithm ends).
\end{corollary}

\begin{proof}
    Under the conditions of the corollary,  lemmas~\ref{lem:3divisible} and~\ref{lem:4large} imply that among the available items, every active agent has at least two divisible special items and at least three items of value above $\frac{1}{2}$. 
    
    If there is an available item $e$ that is divisible and special for two agents, then the two-agent step applies, with $e$ serving as the shareable item. Each of the two agents can get a value of $\frac{1-\rho}{2}$ from $e$, and in addition a value of at least $\frac{1}{2}$ from an item of such value. For $\rho \le \frac{2}{3}$, the total value is at least $\rho$, as desired.
    
    If no available item is divisible and special for two agents, then every agent can get two such items, implementing the 2-matching step.
\end{proof}

This completes the proof of Theorem~\ref{thm:main}.

\section{Discussion}

Our allocation algorithm runs in polynomial time if the MMS values of agents are known, as then scaling the valuations so that the MMS value is~1 can be done in polynomial time. Computing MMS values is NP-hard, even when all items are indivisible. If MMS values are not known, then adaptation of known techniques (such as those in Section 2.4 of~\cite{BUF23}) to our setting give an FPTAS, producing for every $\epsilon > 0$ a $(\frac{3}{5} - \epsilon)$-MMS allocation in time polynomial in $n$, $m$ and $\frac{1}{\epsilon}$.

We showed that our allocation algorithm outputs $\frac{3}{5}$-MMS allocations. It is known that for every $n \ge 2$, there are allocation instances in which no allocation gives every agent more than $\frac{2}{3}$-MMS~\cite{BeiLL25}, and it is open whether $\frac{2}{3}$-MMS can always be achieved, For $n \le 4$, it can~\cite{BeiDLW26}. We remark that also our algorithm outputs $\frac{2}{3}$-MMS allocations when $n \le 4$. (We sketch the proof. The obstacle for our algorithm to achieving  $\frac{2}{3}$-MMS is the two-agent steps, when the shareable item is indivisible. However, for $n = 4$, after executing a two-agent step, only two agents remain, and a single lone pure divider step completes the allocation.)

For $n > 4$ our allocation algorithm (as stated) does not guarantee $\frac{2}{3}$-MMS. 
For every $\epsilon > 0$ and $\rho > \frac{3}{5} + \epsilon$, we present allocation instances with suitably large $n$ in which our algorithm fails to  give some agent an acceptable bundle. 

For a large integer $k$ and $\epsilon = \frac{1}{5k}$, suppose that there are $5k$ items and $3k + 1$ agents. Each agent values each of the items at $\frac{3}{5} + \epsilon$, views $3k$ items as indivisible, and $2k$ items as divisible. For the first $2k$ agents, the first $2k$ items are divisible, but for the remaining $k+1$ agent, the last $2k$ items are divisible. The MMS of every agent is~1: the $3k$ indivisible items are in separate MMS parts, and the divisible items have sufficient value to complete each of these parts to value~1, and to create one additional part of value~1. 

Suppose that $\rho > \frac{3}{5} + \epsilon$. Then the only steps applicable for our algorithm are the two-agent steps. Suppose that the first $2k$ agents perform two-agent steps, and in the process each of them takes an item that is divisible for the remaining agents. Then, we are left with $k+1$ agents and $2k$ indivisible items. Thus, some agent will get at most one item, for a value of $\frac{3}{5} + \epsilon$.

The above is not a conclusive negative example for our allocation algorithm, as it assumes an ``adversarial" execution of our algorithm. It could be that an ``angelic" execution of the algorithm always offer an approximation ratio better than $\frac{3}{5}$, perhaps even $\frac{2}{3}$. In fact, in the above negative example all items have the same value, and it is known that in such instances, $\frac{2}{3}$-MMS allocations exist~\cite{BeiDLW26}.

When setting $\frac{3}{5} < \rho \le \frac{2}{3}$, what may an angelic execution do differently than an adversarial execution? For two-agent steps, if there are several available options for such a step, it may choose the ``best" one (the ratio is at least $\frac{3}{5}$ even if the choice among these options is the ``worse" one). Moreover, it may interleave two-agent steps and lone pure divider steps in the ``best" order. (In our description of the algorithm two-agent steps have priority over lone pure divider steps, but the ratio is at least $\frac{3}{5}$ regardless of how these two types of steps are interleaved.)  Finally, if a 2-matching step becomes applicable, it is executed immediately, even if other types of steps are applicable.

\subsection*{Acknowledgments}

This research was supported in part by the Israel Science Foundation (grant No. 1122/22).

\bibliographystyle{alpha}


\end{document}